\documentclass[a4paper,10pt]{article}
\usepackage[utf8x]{inputenc}
\usepackage{graphicx,latexsym}
\usepackage{amssymb,amsthm,amsmath}
\usepackage{booktabs}
\usepackage{multirow}
\usepackage{url}
\usepackage[numbers]{natbib}
\usepackage{amsfonts}
\usepackage{subfig}
\usepackage{wrapfig}
\usepackage{float}

\newcommand{\beq}{\begin{equation}}
\newcommand{\eeq}{\end{equation}}
\newcommand{\beqnl}{\begin{equation*}}
\newcommand{\eeqnl}{\end{equation*}}

\newcommand{\sof}[1]{\left[#1\right]}

\newcommand{\set}[1]{\left\lbrace #1\right\rbrace}
\newcommand{\of}[1]{\left(#1\right)}

\newcommand{\piecewise}[4]{
\ensuremath{\left\{
   \begin{array}{l l}
   #1 & \quad #2\\
   #3 & \quad #4\\
   \end{array} \right.
}
}

\newcommand{\Prob}[1]{\text{Pr}\sof{\,#1\,}}
\newcommand{\E}[1]{\mathbb{E}\sof{#1}}

\newcommand{\floor}[1]{\lfloor #1 \rfloor}
\newcommand{\ceiling}[1]{\left\lceil #1 \right\rceil}

\newcommand{\R}{\ensuremath{\mathbb{R}}}

\newcommand{\N}{\ensuremath{\mathbb{N}}}

\newcommand{\comment}[1]{}
\newcommand{\mc}[1]{\ensuremath{\mathcal{#1}}}

\newcommand{\suchthat}{\,:\,}
\newcommand{\given}{\,|\,}

\newcommand{\emdash}{\phantom{.}---\phantom{.}} 

\newcommand{\etal}{\textit{et al.}}
\newcommand{\graph}{\mathcal{H}_{n,p}}

\newcommand{\RoutingTime}{\mathcal{T}}

\newcommand{\lambdaMax}{\lambda}
\newcommand{\bestRoute}[2]{\Prob{#1 \succ #2}}

\newcommand{\defn}[1]{\textit{\textbf{#1}}}
\newtheorem{theorem}{Theorem}
\newtheorem{corollary}{Corollary}

\newtheorem{lemma}{Lemma}

\title{Optimal Degree Distributions for Uniform Small World Rings}
\author{
R. Seth Terashima\footnote{{\tt sethterashima@gmail.com}, contact author.}
\and James D. Fix
\footnote{{\tt jimfix@reed.edu}, Mathematics Department, Reed College,
3203 S.E. Woodstock Blvd. Portland, OR 97202}}
\begin{document}
\maketitle
\begin{abstract}
Motivated by Kleinberg's \citep{kleinberg_lattice} and subsequent work, we consider the performance of greedy routing on a directed ring of $n$ nodes augmented with long-range contacts.
In this model, each node $u$ is given an additional $D_u$ edges, a degree chosen from a specified probability distribution.
Each such edge from $u$ is linked to a random node at distance $r$ ahead in the ring with probability proportional to $1/r$, a ``harmonic" distance distribution of contacts.
Aspnes et al.\, \citep{aspnes_fault_tolerant} have shown an $O(\log^2 n / \ell)$ bound on the expected length of greedy routes in the case when each node is assigned exactly $\ell$ contacts and, as a consequence of recent work by Dietzfelbinger and Woelfel \citep{ring_lower_bound_tight}, this bound is known to be tight.
In this paper, we generalize Aspnes' upper bound to show that {\em any} degree distribution with mean $\ell$ and maximum value $O(\log n)$ has greedy routes of expected length $O(\log^2n / \ell)$, implying that any harmonic ring in this family is asymptotically optimal.
Furthermore, for a more general family of rings, we show that a fixed degree distribution is optimal.
More precisely, if each random contact is chosen at distance $r$ with a probability that decreases with $r$, then among degree  distributions with mean $\ell$, greedy routing time is smallest when every node is assigned $\floor{\ell}$ or $\ceiling{\ell}$ contacts.
\end{abstract}

\section{Introduction}
\comment{
Here, we take on the study of greedy routing in Kleinberg's lattice graphs \citep{kleinberg_lattice}; in particular, we consider certain directed ring models studied by Aspnes \etal\  \citep{aspnes_fault_tolerant} and Barri\`{e}re \etal\  \citep{ring_model}, among others.
Our main contribution, built directly off their work, is an upper bound on greedy routing time as a function of network size and node out-degree, one that is independent of how node out-degrees are distributed and that nearly matches the lower bound of Giakkoupis and Hadzilacos \citep{ring_lower_bound}.
The end result is an asymptotically precise description of expected routing time in these networks.
}
\subsection{Background} 
Our work extends results that lie at the intersection of mathematically modeling the small world phenomenon in social networks and the design of decentralized peer-to-peer networks.
In both contexts, a central problem is determining how efficiently a message can be routed between arbitrary nodes of a network.

The notion of a \defn{small world} is most frequently encountered in the context of social networks.
The term refers to systems where entities are highly clustered and linked to only a small portion of the network, but are nevertheless connected by short paths. 
Research, notably the letter-forwarding experiments conducted by Stanley Milgram in the 1960s \citep{milgram}, suggests that small world networks exist in the real world.
The work of Kleinberg \citep{kleinberg_lattice} and others (see Kleinberg \citep{kleinberg_review} for a review) provides insight into conditions under which people can efficiently \emph{find} short paths using only local information, as modeled by, say, greedy routing.

Kleinberg's model begins with an $n$-by-$n$ lattice of nodes.
Each node is connected to all other nodes within a specified distance.
Additionally, each node is given $\ell$ \defn{long-range contacts} (or \textit{LRCs}) chosen according to some stochastic process.
Kleinberg considered power-law distributions, in which the probability that a node $u$ chooses node $v$ as an LRC is proportional to $\delta^{-\beta}$, where $\beta$ is a constant and $\delta$ is the distance from $u$ to $v$.
For $\beta = 2$ and $\ell = 1$, he showed greedy routing takes $O(\log^2 n)$ (this bound is tight \citep{analyzing_kleinberg}), whereas for $\beta \neq 2$ greedy routing time is bounded below by a polynomial in $n$. In general, the optimal value for $\beta$ is equal to the dimension of the lattice.

Subsequent work has instead considered a ring model. 
Barri\`{e}re \etal\ \citep{ring_model} showed that in this variation, $\beta = 1$ and $\ell = 1$ allows $O(\log^2 n)$ routing time.
Aspnes, Diamadi, and Shah \citep{aspnes_fault_tolerant} generalized this to $O(\log^2 n / \ell)$ as part of a proposed P2P network.
Their system bears many similarities to Chord \citep{chord}, a system for maintaining a DHT which provides $\Theta(\log n)$ routing time using a ring-based overlay network with $\log_2 n$ LRCs per node.

In the context of both social and computer networks (particularly those designed with fault-tolerance in mind), it makes sense to consider graphs in which nodes have a random number of LRCs. 
Fraigniaud and Giakkoupis \citep{power_law_distributions} studied the effect of power-law LRC-degree distributions on the ring-based model. 
(We distinguish between LRC-degree distributions, which control the number of LRCs assigned to a node, and LRC-distance distributions, which dictate how those nodes are chosen). 
In particular, they consider a family of zeta distributions, modified to hold the mean at two regardless of the power-law exponent. 
For directed graphs, greedy routing performs in $O(\log^2 n)$ time, while for undirected graphs, routing time depends critically on the power-law exponent.

Work on the corresponding lower bounds considers a broader class of graphs. 
In this model, each node is randomly assigned a set $D \subset \set{1, \ldots, n}$, which contains the distances to that node's LRCs. 
(This allows random graphs unobtainable with independent LRC-degree and LRC-distance distributions).
This process is uniform\emdash the distribution used to choose $D$ is the same for all nodes. 
Giakkoupis and Hadzilacos \citep{ring_lower_bound} gave an $\Omega(\log^2 n / \E{|D|} a^{\log^* n})$ bound on the average expected routing time (where $a > 1$ is a constant), which was later improved to $\Omega(\log^2 n / \E{|D|})$ by Dietzfelbinger and Woelfel \citep{ring_lower_bound_tight}.

\subsection{Statement of results}
The $\Omega(\log^2 n/ \E{|D|})$ lower bound is tight in the sense that the model under consideration permits distributions resulting in $O(\log^2 n/ \E{|D|})$ routing time, such as those studied by Aspnes (if $\ell$ LRCs are chosen with replacement, then $\ell \geq \E{|D|}$).                                                                  
However, establishing upper bounds for different distributions remains an open problem. In this paper, we consider the ring model (with a harmonic LRC-distance distribution). We show that if the LRC-degree distribution has mean $\ell$ and the property that no node can have more than $O(\log n)$ LRCs, then the expected routing time between any two nodes is $O(\log^2 n / \ell)$ (Theorem~\ref{thm:general_dist_upper_bounds}). Hence, this sub-family of graphs provides asymptotically optimal routing time. 

Finally, fixing the mean degree, we investigate what LRC-degree distributions optimize greedy routing performance. We give Theorem \ref{thm:equal_is_best}, whose lemmata establish that gaining contacts provides limited returns on the expected length of each greedy hop. This holds for any LRC-distance distribution under which closer nodes are more likely to be selected as LRCs than those farther away. Thus, greedy routing in this family of directed graphs is optimized when LRC-degrees do not vary.

\section{Model Description}
Let $\mc{R}_n = (V, E)$ be the directed ring graph with $n$ vertices,
which we identify with the integers: 
\beqnl 
	V = \set{0, \ldots, n - 1}, \quad E = \set{(u, u + 1) \suchthat u \in V}. 
\eeqnl 
All operations on vertices are performed modulo $n$.

Define the function $\delta : V \times V \rightarrow \N$ to be the distance from $u$ to $v$ along the ring:
\beqnl
	\delta(u, v) = \piecewise{v - u}{\text{if } v \geq u}{n - (u-v)}{\text{if } v < u.}
\eeqnl

We wish to construct an augmented graph containing $\mathcal{R}_n$, but where each node has some number of additional out-going edges according to a specified distribution. Let $p(n, \cdot)$ be a probability distribution on $\N$ (that is, there is a different distribution for each value of $n$). With each node $u$ of $\mathcal{R}_n$, associate a random variable $D_u$ taken from this distribution: $\Prob{D_u = k} = p(n, k)$. This variable indicates how many additional edges will be attached to $u$ (since these edges will be chosen with replacement, they will not in general be distinct). In the future, we will write $p(n, k)$ as $p(k)$, with the dependence on $n$ made implicit.

Given $u \in V$ and $j \in \N$, let $\Delta_{u,j} \in \set{1, \ldots, n - 1}$ be a random variable such that
\beqnl
	\Prob{\Delta_{u,j} = r} \propto \frac{1}{r}.
\eeqnl
Note that the proportionality constant is the reciprocal of the $(n-1)^\text{th}$ harmonic number: $H_{n-1}^{-1} = \of{\sum_{i=1}^{n-1} 1/i}^{-1} = \Theta(1 / \log n)$.

Define $E_u = \set{(u, u + \Delta_{u,j}) \suchthat 1 \leq j \leq D_u}$, and let $E' = E \cup \bigcup_{u \in V} E_u$. The graph $\graph = (V, E')$ so constructed is a \defn{harmonic ring}. Given $u \in V$, let $C_u = \set{v \in V : (u, v) \in E_u}$. Elements of $C_u$ are \defn{long-range contacts} (LRCs) of $u$.

For $u, v \in V$ and $A \subset V$, let $\Prob{u \rightarrow v}$ be the probability that $(u, v) \in E'$, and let $\Prob{u \rightarrow A}$ be the probability that there exists a node $w$ with $(u, w) \in E'$.

We now introduce some notation to formalize the notion of greedy routing. If $u$ and $v$ are nodes of $\graph$, a \defn{greedy route} from $u$ to $v$ is a sequence $u = s_0, s_1, \ldots, s_k = v$ such that $(s_j, s_{j + 1}) \in E'$ and if $(s_j, w) \in E'$, then $\delta(s_{j + 1}, v) \leq \delta(w, v)$. Since $(s_j, s_j + 1) \in E'$, we can always make progress towards $v$; a greedy route exists between arbitrary vertices. Because $\delta(\cdot, v)$ is injective, the greedy route is unique. The \defn{greedy routing time} from $u$ to $v$ is $k$. This definition formalizes the notion of always taking the route that looks best from a limited, local perspective: each node ``knows'' (has links to) a limited number of other nodes, and always passes a message along to the one closest to the destination.

Finally, let $T_r$ denote the expected greedy routing time when the distance between the source and destination nodes is $r$, and define $\RoutingTime_{n,p}$ to be the average expected routing time between all pairs of nodes in $\graph$:
\beqnl
	\RoutingTime_{n,p} = \frac{1}{n} \sum_{r = 0}^{n - 1} T_r.
\eeqnl

\section{Routing complexity}
Our upper bound proof follows the same basic outline as Kleinberg's original argument: we first find a bound on the expected time it takes to cut an initial distance in half, and then couple this with the observation that this must be done at most $\log_2 n$ times. 

\begin{lemma}
	\label{lem:halfway_bound}
	Let $\graph = (V, E)$ be a harmonic ring. Let $u, v \in V$ be distinct, and let $B = \set{w \in V \mid \delta(w, v) \leq \delta(u, v)/2 }$. Then $\Prob{u \rightarrow B \mid D_u = 1} = \Theta\of{\frac{1}{\log n}}$.
\end{lemma}
\begin{proof}
	Assume without loss of generality that $u = 0$. Then
	\beqnl
		\begin{split}
		\Prob{u \rightarrow B \mid D_u = 1} &= \sum_{w \in B} \Prob{u \rightarrow w \mid D_u = 1}\\
		&= H_{n-1}^{-1} \sum_{w \in B} \frac{1}{\delta(u, w)}\\
		&= H_{n-1}^{-1} \sum_{r = v / 2}^{v} \frac{1}{r}.
		\end{split}
	\eeqnl
	Since $1/r$ is a decreasing function,
	\beqnl
		\int_{v/2}^{v} \frac{dr}{r} < \sum_{r = v/2}^{v} \frac{1}{r} < \frac{1}{v/2} + \int_{v/2}^{v+1} \frac{dr}{r};
	\eeqnl
	that is,
	\beqnl
		H_{n-1}^{-1} \log 2 < \Prob{u \rightarrow B \mid D_u = 1} < H_{n-1}^{-1} \of{2 + \log 4}.
	\eeqnl
	Hence, $\Prob{u \rightarrow B \mid D_u = 1} = \Theta(H_{n-1}) = \Theta(1 / \log n)$.
\end{proof}

Lemma~\ref{lem:halfway_bound} makes it easy to work with the probability of cutting the remaining distance in half. We will now take advantage of this to formulate and solve a recurrence describing how long greedy routing takes.

\begin{theorem}
	\label{thm:general_dist_upper_bounds}
	Let $\graph$ be a harmonic ring. Let $X$ be a random variable taken from the distribution $p(n,\cdot)$, and let $c > 0$ be a constant such that for all $n$, $\Prob{X \leq c \log n} > 0$.
	Then
	\beqnl
		\RoutingTime_{n,p} = O\of{\frac{\log^2 n}{\E{X \mid X \leq c \log n}}}.
	\eeqnl
\end{theorem}

\begin{proof}
	We will prove that this upper bound holds for the expected routing time between arbitrary source-target pairs, Consider the greedy route from $u$ to $v$. How many steps does it take to cut the initial distance in half? We found an answer to this question under the assumption that each node had a single LRC, but now require a more general result. As before, define $B = \set{w \in V \mid \delta(w, v) \leq \delta(u, v) / 2}$. The probability that $u$ has an LRC in $B$ is the probability that \textit{not all} of $u$'s contacts \textit{miss} $B$:
	\beqnl
	\begin{split}
		\Prob{u \rightarrow B} &= 
			\sum_{d=0}^{\infty}p(d)\of{1 - \Prob{u \not\rightarrow B \mid D_u = d}}\\
			&= 1 - \sum_{d=0}^{\infty}p(d)\of{1 - \Prob{u \rightarrow B \mid D_u = 1}}^{d}.
	\end{split}
	\eeqnl

	The probability that $u$ is linked to a node in $B$ is at least $\Prob{u \rightarrow B}$, since the latter value does not account for the $(u, u+1)$ edge. Furthermore, the closer a message gets to $B$, the greater its chances of entering $B$ on the next step; that is, $\delta(w, v) < \delta(w', v)$ implies $\Prob{w \rightarrow B} > \Prob{w' \rightarrow B}$. This follows from the fact that $\Prob{w \rightarrow v'} > \Prob{w' \rightarrow v'}$ for all $v' \in B$. Therefore if $s_j$ is on the greedy route from $u$ to $v$, $\Prob{s_j \rightarrow B} \geq \Prob{u \rightarrow B}$.

	If $s_0, \ldots s_k$ is the greedy route from $u$ to $v$, let $M$ be the random variable defined by $M = \min\set{j \suchthat s_j \in B}$. We have
	\beqnl
		\E{M} < \frac{1}{\Prob{u \rightarrow B}} = 
			\frac{1}{1 - \sum_{d=0}^{\infty}p(d)\of{1 - \Prob{u \rightarrow B \mid D_u = 1}}^{d}}.
	\eeqnl

	Since $\Prob{u \rightarrow B \mid D_u = 1} = \Theta(1/\log n)$, it follows that there exists some positive constant $\beta$ such that for all sufficiently large $n$, $\Prob{u \rightarrow B \mid D_u = 1} > \beta / \log n$. Let $x = 1 - \beta / \log n$ (although $x$ depends on $n$ we will refrain from adding a subscript, so as to avoid clutter). In other words, $x$ is an upper bound for the probability that a given LRC \textit{fails} to cut the remaining distance in half. Hence, for large $n$,
	\beqnl
		\E{M} < \frac{1}{1 - \sum_{d = 1}^{\infty}p(d)x^d} \overset{\text{call}}{=} \lambdaMax.
	\eeqnl
	The value of $\lambdaMax$ is independent of $u$ and $v$. Therefore $\lambdaMax$ is an upper bound for the expected time it takes to cut the remaining distance in half between \textit{any} two nodes in $\graph$. Hence,
	\beqnl
		T_r < \lambdaMax+\max\set{T_s \suchthat s \leq r/2}.
	\eeqnl
	Since $T_0 = 0$, this yields:
	\beqnl
		T_r < \lambdaMax \log_2 r.
	\eeqnl
	Therefore $\lambdaMax \log_2 n$ is an upper bound for the expected routing time between any two vertices (and hence is an upper bound for the average expected routing time over all pairs of vertices). Thus
	\beq
		\label{eq:upper_bound}
		\RoutingTime_{n,p} < \frac{\log_2 n}{{1 - \sum_{d = 0}^{\infty}p(d)x^d}}.
	\eeq

	Let $L = \floor{c \log n}$, and define the probability distribution $q$ by:
	\beqnl
		q(d) = \left\{
			\begin{array}{l l}
				p(d) / P & \quad \text{if $d \leq L$}\\
				0 & \quad \text{otherwise}\\
			\end{array} \right.,
	\eeqnl
	where $P = \Prob{X \leq L}$. Let $Y$ be a random variable taken from the distribution $q$. Then $\E{X \mid X \leq c \log n} = \E{Y}$. Define the function $A : \N \rightarrow \R$ by
	\beqnl
		A(n) = 1 - \sum_{d=0}^{\infty}p(d)x^{d}.
	\eeqnl

	That is, $A(n)$ is the expression appearing in the denominator of (\ref{eq:upper_bound}). It suffices to show that $A(n) = \Omega\of{\frac{\E{Y}}{\log n}}$.

	We have:
	\beqnl
	A(n) = \sum_{d = 0}^{\infty} p(d)\of{1 - x^d} = (1-x)\sum_{d=0}^{\infty}p(d) \of{1 + x + \cdots + x^{d-1}}.
	\eeqnl
	Let $f: \N \to \R$ be the function $f(d) = \sum_{i=0}^{d-1}x^i$. Then
	\begin{gather*}
	\sum_{d=L+1}^{\infty}p(d)f(d)
		> f(L)\sum_{d=L+1}^{\infty}p(d)
		= f(L)(1 - P)
		= f(L)(1/P - 1)\sum_{d=0}^L p(d)\\
		> (1/P - 1)\sum_{d=0}^L p(d)f(d)
		= \sum_{d=0}^L (q(d) - p(d))f(d)
		= \sum_{d=0}^L q(d)f(d) - \sum_{d=0}^L p(d)f(d).
	\end{gather*}
	Hence,
	\beqnl
		\sum_{d=0}^{\infty}p(d)f(d) \geq \sum_{d=0}^{L}q(d)f(d).
	\eeqnl

	We know that $0 < x < 1$, so whenever $1 \leq d \leq L$,
	\beqnl
		f(d) = 1 + x + \cdots + x^{d - 1} > d x^{d-1} > d x^{L}.
	\eeqnl
	Returning to our expression for $A(n)$ and noting that $x^{L} = \of{1 - \beta/\log n}^{c \log n}$ converges to a constant as $n$ grows large,
	\beqnl
		A(n) \geq (1 - x)x^{L}\sum_{d=1}^{L} d q(d) = (1 - x)x^{L}\E{Y} = \Omega\of{\frac{\E{Y}}{\log n}}.
	\eeqnl
	This concludes the proof.
\end{proof}

If the maximum possible number of LRCs that can be assigned to a particular node is $O(\log n)$, the result becomes much cleaner.

\begin{corollary}
	Let $\graph$ be a harmonic ring where $p(n, \cdot)$ has mean $\ell$. Then if there is some constant $c > 0$ such that $p(n, d) = 0$ whenever $d \geq c \log n$, then $\RoutingTime_{n,p} = O(\log^2 n / \ell)$.\qed
\end{corollary}

This bound is tight \citep{ring_lower_bound_tight}.

\section{Optimal LRC-degree distributions}

The previous results demonstrate that the asymptotic performance of greedy routing depends almost entirely on the mean of the distribution used to choose the number of LRCs for each node. Experimentally, however, different distributions can result in significantly different average routing times. In this section, we prove that of those distributions with mean $\ell$, greedy routing is optimized when every node has $\floor{\ell}$ or $\ceiling{\ell}$ LRCs. This result holds not just for harmonic rings, but in any variant where a closer node is more likely to be selected as an LRC than one farther away.

Let $\bestRoute{j}{i}$ be the probability that a node at distance $j$ from the destination routes to a node at distance $i$ from the destination. Using this notation,
\beqnl
	T_r = 1 + \sum_{s = 0}^{r - 1} \bestRoute{r}{s} T_s \quad (r > 0).
\eeqnl

\begin{lemma}
	\label{lem:shorter_is_better}
	$T_r$ is an increasing function of $r$.
\end{lemma}
\begin{proof}
	Let $J_r$ be a random variable such that $\Prob{J_r = s} = \bestRoute{r}{r-s}$. For $r \geq 1$, define $\tau_r = T_r - T_{r - 1}$. Given $a < n$, assume that $\tau_i > 0$ whenever $i < a$. Then
	\beqnl
	\begin{split}
		\tau_a &= \sum_{r = 1}^{a - 1} \bestRoute{a}{r} T_r - \sum_{r = 1}^{a-2} \bestRoute{a - 1}{r} T_r\\
			&= \sum_{r=1}^{a-1} \Prob{J_a \leq r} \tau_{a - r}
				- \sum_{r=1}^{a-2} \Prob{J_{a-2} \leq r} \tau_{(a - 1) - r }\\
			&= \sum_{r=1}^{a-1} \Prob{J_a \leq r} \tau_{a - r}
				- \sum_{r=2}^{a-1} \Prob{J_{a - 1} \leq r-1} \tau_{a - r}\\
			&> \sum_{r = 2}^{a - 1} \of{\Prob{J_a \leq r}-\Prob{J_{a-1} \leq r - 1}}\tau_{a-r}\\
			&> 0
	\end{split}
	\eeqnl
	The first inequality results from the fact that closer nodes are more likely to be chosen as LRCs than those farther away (this is a sufficient condition for the proof to work). The lemma follows by induction.
\end{proof}

\begin{lemma}
	\label{lem:diminishing_returns}
	Let $p$ be a distribution on $\N$ with mean $\mu$, and let $f : \R \rightarrow \R$ be a twice-differentiable function with $f(x) \geq 0$, $f'(x) < 0$, and $f''(x) > 0$. Then $\sum p(d)f(d)$ is smallest when the support of $p$ is $\set{\floor{\mu}, \ceiling{\mu}}$.
\end{lemma}

This lemma, the proof of which will be omitted, makes a simple statement about optimizing the expected value of a function that provides diminishing returns. When $a < \mu < b$, the benefit of increasing $p(\floor{\mu})$ (while decreasing $p(a)$) is greater than corresponding the cost of increasing $p(\ceiling{\mu})$ (while decreasing $p(b)$).

When considering what benefit might be obtained for greedy routing by varying the LRC-degree distribution, 
we find that expected route length is governed by this lemma.  That is, roughly speaking, a node gets diminishing returns on the expected jump lengths it can provide with each additional LRC it is allocated. The following theorem argues that since longer jumps are always better (Lemma~\ref{lem:shorter_is_better}), the best thing to do is to ensure that LRC-degree selection varies as little as possible (Lemma~\ref{lem:diminishing_returns}).

\begin{theorem}
	\label{thm:equal_is_best}
	Let $S_{\ell}$ be the set of probability distributions on $\N$ with mean $\ell \in \N$. Let $p \in S_{\ell}$ be the distribution with support $\set{\floor{\ell}, \ceiling{\ell}}$. Then for all $q \in S_{\ell}$, $\RoutingTime_{n,p} \leq \RoutingTime_{n,q}$.
\end{theorem}
\begin{proof}

	Consider two arbitrary nodes, $u$ and $v$. We will show that the expected routing time from $u$ to $v$ is smallest when $D_u$ is chosen according to $p$, and that this is true regardless of what distribution is used to choose $D_{u+1}, \ldots, D_v$ (as long as the same distribution is used for all them).\footnote{This theorem is false if different nodes are assigned LRC-degrees based on different distributions; for example, an 8-node harmonic ring with 8 total LRCs is optimal when nodes alternate between two and zero LRCs.}
	
	So assume that $D_{u+1}, \ldots, D_{v}$ are chosen from the same distribution (keeping Lemma~\ref{lem:shorter_is_better} applicable). As before, let $\tau_i = T_i - T_{i-1}$ (here we will restrict the definitions of $T_i$ and $\tau_i$ to refer only to greedy paths where the destination node is $v$); by Lemma~\ref{lem:shorter_is_better}, $\tau_i > 0$. Let $\Delta$ be a random variable such that $\Prob{\Delta = r}$ is equal to the probability that $u$ routes to $u + r$. Define $T_r(d)$ to be $T_r$ given that the source node, $u$, has been assigned $d$ LRCs.
	\beqnl
	\begin{split}
		T_r(d) &= 1 + \sum_{s=1}^{r} \Prob{\Delta = s \given D_u = d} T_{r - s}\\
			&= 1 + \sum_{s=1}^{r} \Prob{\Delta \leq s \given D_u = d} \tau_{r - s}\\
			&= 1 + \sum_{s=1}^{r} \Prob{\Delta \leq s \given D_u = 1}^d \tau_{r - s}.
	\end{split}
	\eeqnl
	This last equality allows us to extend the definition of $T_r(d)$ to include all $d \in \R$. Letting $\alpha_r = \Prob{\Delta \leq r \given D_u = 1}$, we have, for all $d$,
	\beqnl
		T_r'(d) = \sum_{s=1}^{r} (\log \alpha) \alpha_r^d \tau_{r - s} < 0
	\eeqnl
	and
	\beqnl
		T_r''(d) = \sum_{s=1}^{r} ( \log^2 \alpha) \alpha_r^d \tau_{r - s} > 0.
	\eeqnl
	By Lemma~\ref{lem:diminishing_returns}, $T_r = \sum q(d)T_r(d)$ is smallest when $q = p$. Hence using $p$ for all nodes simultaneously minimizes routing times over all distances.
\end{proof}

\bibliographystyle{plain}

\bibliography{seth-rings}
\end{document}